\documentclass{elsarticle}

\usepackage{amsmath}
\usepackage{amssymb}
\usepackage{graphicx}
\usepackage{float}
\usepackage{caption}
\usepackage{amsthm}
\usepackage{color}
\definecolor{green}{rgb}{0,0.5977,0}
\usepackage[linesnumbered,algoruled,boxed,lined]{algorithm2e}

\newcommand{\ee}{\mathbb E}

\newcommand{\abs}[1]{\left|{#1}\right|}

\newcommand{\genseq}[3]{{#1}_1 {#3} {#1}_2 {#3} \dots {#3} {#1}_{#2}}
\newcommand{\seq}[2]{\genseq{#1}{#2}{,}}

\newcommand{\twocases}[4]{\begin{cases} #2 & #1 \\ #4 & #3 \end{cases}}

\newcommand{\txt}[1]{\text{#1}}

\newcommand{\stext}[1]{\ \ \ \ \ \text{(#1)}}

\newcommand{\ipncm}[3]{\begin{figure}[H]\begin{center}\includegraphics[scale = {#1}]{#2.pdf}\caption{#3}\end{center}\end{figure}}

\makeatletter 
\g@addto@macro{\@algocf@init}{\SetKwInOut{Parameter}{Parameters}} 
\makeatother

\usepackage{makecell,url}
\newcommand{\cut}{\textbf{cut}\xspace}
\newcommand{\eval}{\textbf{eval}\xspace}
\theoremstyle{plain}
\newtheorem{theorem}{Theorem}
\newtheorem{lemma}[theorem]{Lemma}

\begin{document}
	\begin{frontmatter}
		\title{Thou Shalt Covet The Average Of Thy Neighbors' Cakes}
		\author{Jamie Tucker-Foltz}
		\ead{jtuckerfoltz@gmail.com}
		\address{Harvard John A. Paulson School Of Engineering And Applied Sciences, Science and Engineering Complex, 150 Western Ave, Boston, MA 02134}
		\begin{abstract}
			We prove an $\Omega(n^2)$ lower bound on the query complexity of \emph{local proportionality} in the Robertson-Webb cake-cutting model. Local proportionality requires that each agent prefer their allocation to the average of their neighbors' allocations in some undirected social network. It is a weaker fairness notion than envy-freeness, which also has query complexity $\Omega(n^2)$, and generally incomparable to proportionality, which has query complexity $\Theta(n \log n)$. This result separates the complexity of local proportionality from that of ordinary proportionality, confirming the intuition that finding a locally proportional allocation is a more difficult computational problem.
		\end{abstract}
		
	\end{frontmatter}
	
	\section{Introduction}\label{secIntro}
	
	In the cake-cutting model of continuous fair division, the gold standard for fairness is \emph{envy-freeness}: each agent $i$ should be allocated a portion of cake such that they prefer their own piece to that of any other agent $j$; in other words, $i$ does not \emph{envy} $j$. Implicit in this definition is the idea that agent $i$ is aware of agent $j$ and their allocation. But what if $i$ and $j$ are complete strangers, and $j$'s allocation has no bearing on $i$'s sense of fairness?
	
	This consideration has motivated the study of ``local'' fairness concepts, where there is an underlying social network, represented as an undirected simple graph $G$, and allocations are only required to be fair with respect to local comparisons between nodes and their neighbors \cite{Original1, Original2}. An allocation is \emph{locally envy-free} if $i$ does not envy any of its neighbors, and \emph{locally proportional} if $i$ receives utility equal to or exceeding their average utility for their neighbors' allocations. Since adding edges to $G$ can only make local envy-freeness more difficult to satisfy, local envy-freeness is a weaker condition than envy-freeness, as the latter is the special case of the former when $G$ is a complete graph. However, the same cannot be said for local proportionality: adding edges to $G$ may actually make it an \emph{easier} condition to satisfy. In the special case where $G$ is a complete graph, local proportionality is equivalent to the well-studied notion of \emph{proportionality}, which simply requires that each of the $n$ agents receive a portion of cake of value at least $\frac{1}{n}$. For a graph $G$ that is not a complete graph or its complement, local proportionality is logically incomparable to proportionality (as illustrated in Section \ref{secFairnessNotions}). The logical relationships between the four fairness notions are illustrated in Figure \ref{figFairnessImplications}. There is also a growing literature on analogous concepts for indivisible goods, which we do not go into; see, for example, Aziz, Bouveret, Caragiannis, Giagkousi, and Lang \cite{Indivisible1}.
	
	\ipncm{.6}{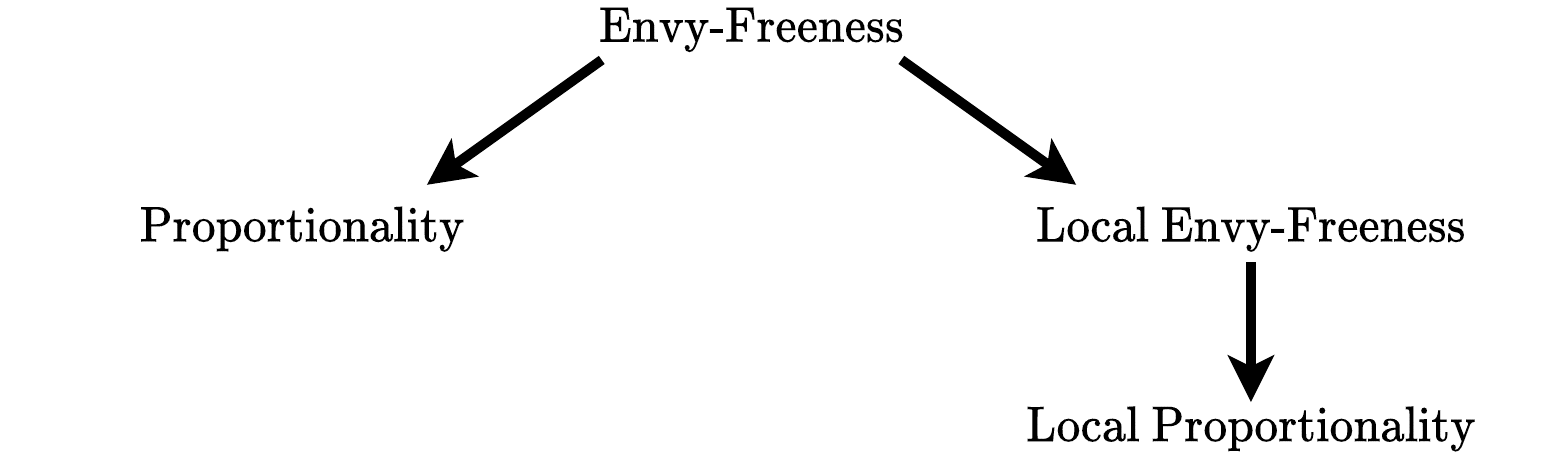}{\label{figFairnessImplications} Diagram showing which fairness properties imply others for a given allocation.}
	
	Since an envy-free allocation is guaranteed to exist, all four fairness notions are always feasible for any input graph $G$. The question is, what is the complexity of computing such an allocation? This question is typically asked in the context of the Robertson-Webb query model \cite{RWModel}, where the valuation functions of the agents are oracles to which an algorithm is allowed to make specific, structured queries. The \emph{query complexity} of a fairness property $\mathcal{F}$ is the minimal number of queries needed to find an allocation satisfying $\mathcal{F}$ in the worst case, as a function of the number of agents $n$.
	
	It is known that the query complexity of proportionality is $\Theta(n \log n)$ \cite{ProportionalityUpperBound, ProportionalityLowerBound}. On the other hand, the asymptotic query complexity of envy-freeness is an open problem, with an $\Omega(n^2)$ lower bound \cite{EFLowerBound} and an $O(n \uparrow \uparrow 6)$ upper bound \cite{EFUpperBound}. The query complexity of local envy-freeness is clearly the same as for envy-freeness, since the worst case is realized when $G$ is a complete graph. In a recent work, Bei, Sun, Wu, Zhang, Zhang, and Zi \cite{Main} give a new algorithm for local proportionality establishing an $O(14^n)$ upper bound. However, prior to this paper, there were no known lower bounds for local proportionality aside from the trivial bound of $\Omega(n \log n)$ obtained from the case where $G$ is a complete graph. Our main result is an improvement of this lower bound to $\Omega(n^2)$. Since proportionality can be guaranteed with only $\Theta(n \log n)$ queries, this result shows that local proportionality is a fundamentally more difficult property to achieve. Table \ref{tabComplexity} lists the query complexities of all four fairness notions, including this new contribution.
	
	\begin{table}[h!]
		\centering
		\begin{tabular}{|l|l|l|} 
			\hline
			Fairness Criterion & Lower Bound & Upper Bound \\
			\hline\hline
			Proportionality & $\Omega(n \log n)$ \cite{ProportionalityLowerBound} & $O(n \log n)$ \cite{ProportionalityUpperBound} \\
			\hline
			Local Proportionality & $\mathbf{\Omega(n^2)}$ \textbf{[this paper]} & $O(14^n)$ \cite{Main} \\
			\hline
			\makecell[l]{Envy-Freeness / \\ Local Envy-Freeness} & $\Omega(n^2)$ \cite{EFLowerBound} & $O\left(n^{n^{n^{n^{n^{n}}}}}\right)$ \cite{EFUpperBound} \\
			\hline
		\end{tabular}
		\caption{Query complexities of computing fair allocations in the Robertson-Webb model.}
		\label{tabComplexity}
	\end{table}
	
	\section{Definitions and example}\label{secFairnessNotions}
	
	In the cake-cutting model, the cake is represented by the unit interval $[0, 1]$. A \emph{piece of cake} is a subset of $[0, 1]$ comprised of a finite union of nonempty closed intervals. An instance of the \emph{graphical cake-cutting problem} is specified by an undirected simple graph $G$ of $n$ vertices, and a valuation function $v_i$ for each agent $i \in V(G)$. We write $N(i)$ to denote the set of neighbors of $i$ in $G$, and $\deg(i) := \abs{N(i)}$. For any piece of cake $X$, $v_i(X)$ is the value that agent $i$ has for $X$. The valuation functions are assumed to satisfy the following standard axioms:
	\begin{itemize}
		\item $v_i(X) \geq 0$ for any piece of cake $X$.
		\item $v_i([0, 1]) = 1$.
		\item {[Additivity]} For disjoint intervals $I_1, I_2 \subseteq [0, 1]$, $v_i(I_1 \cup I_2) = v_i(I_1) + v_i(I_2)$.
		\item {[Divisibility]} For any interval $I \subseteq [0, 1]$ and $0 \leq \lambda \leq 1$, there exists an interval $I' \subseteq I$ such that $v_i(I') = \lambda v_i(I)$.
	\end{itemize}

	An \emph{allocation} is a collection of pieces of cake $\{A_i\}_{i \in V(G)}$ whose union is $[0, 1]$, where, for any $i \neq j$, $A_i \cap A_j$ is a set of measure zero (i.e., the $A_i$ pieces only overlap at endpoints of intervals). For technical convenience, we assume without loss of generality each of the intervals comprising each $A_i$ must have positive measure, since singleton points do not matter anyway. An allocation is:
	\begin{itemize}
		\item \emph{proportional} if, for all $i \in V(G)$,
		$v_i(A_i) \geq \frac{1}{n}.$
		\item \emph{locally proportional} if, for all $i \in V(G)$,
		$v_i(A_i) \geq \frac{1}{\deg(i)} \sum_{j \in N(i)} v_i(A_j).$
		\item \emph{envy-free} if, for all $i, j \in V(G)$,
		$v_i(A_i) \geq v_i(A_j).$
		\item \emph{locally envy-free} if, for all $\{i, j\} \in E(G)$,
		$v_i(A_i) \geq v_i(A_j).$
	\end{itemize}

	\ipncm{.6}{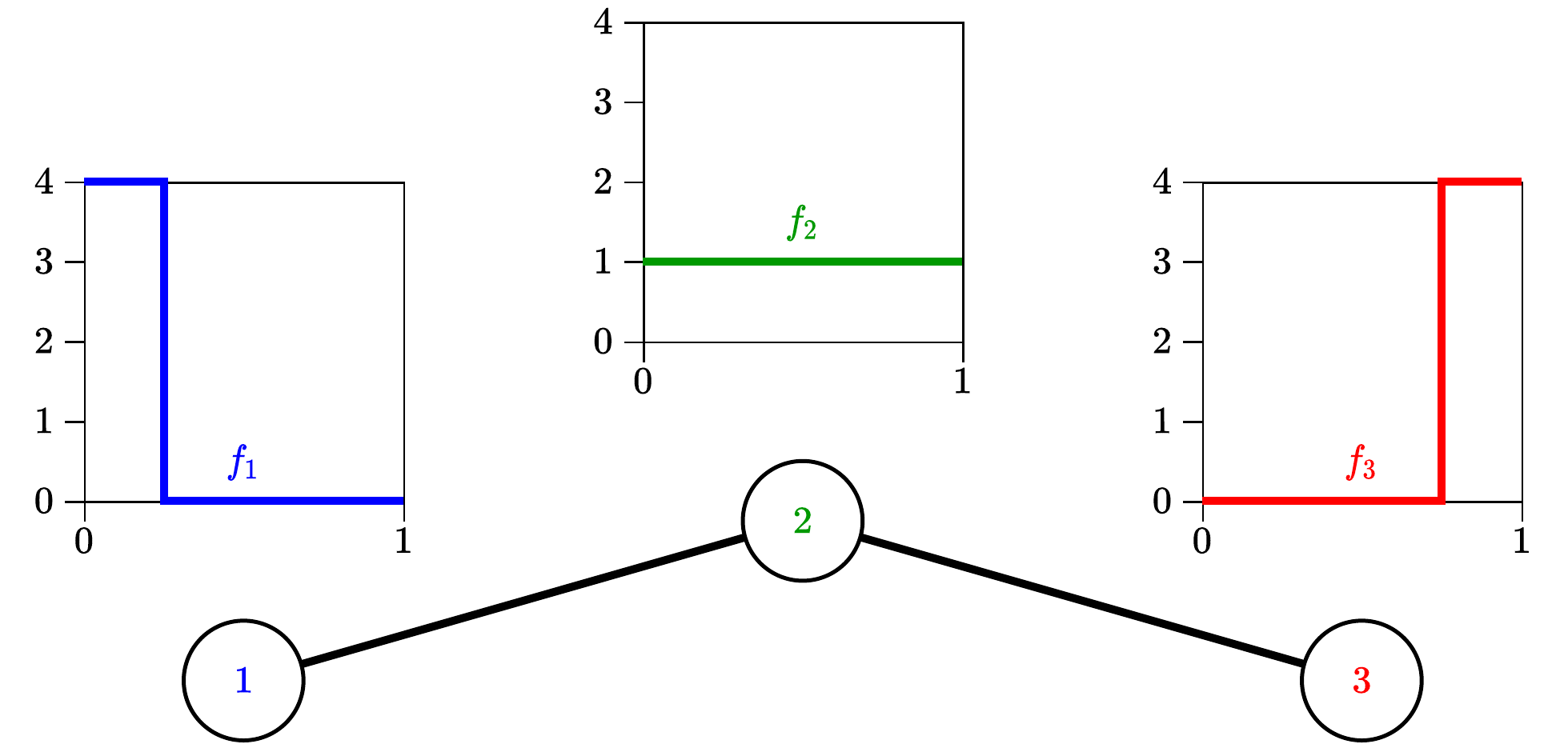}{\label{figExample}An example instance to the graphical cake-cutting problem where $G$ is a path on 3 vertices. Above each vertex $i$ is shown a density function $f_i$ representing the value agent $i$ has for the different parts of the cake, where, for any piece of cake $X$, $v_i(X) := \int_{x \in X} f_i(x)$. In words, agent 2 values all parts of the cake equally, whereas agents 1 and 3 only value the leftmost and rightmost quarters of the cake, and uniformly within their respective quarters.}

	For example, consider the instance of the graphical cake-cutting problem in Figure \ref{figExample}, and the following allocations:
	\begin{itemize}
		\item $\{A_1 = [0, \frac14], A_2 = [\frac14, \frac34], A_3 = [\frac34, 1]\}$ is envy-free, and thus locally envy-free, proportional, and locally proportional as well.
		\item $\{A_1 = [0, \frac{1}{16}], A_2 = [\frac14, \frac34], A_3 = [\frac{1}{16}, \frac14] \cup [\frac34, 1]\}$ is locally envy-free but not proportional, since agent 1 receives only $\frac14$ of the cake according to $v_1$.
		\item $\{A_1 = [0, \frac{1}{16}], A_2 = [\frac14, \frac{11}{16}], A_3 = [\frac{1}{16}, \frac14] \cup [\frac{11}{16}, 1]\}$ is locally proportional, but it is not proportional since agent 1 receives only $\frac14$ of the cake according to $v_1$, and not locally envy-free since agent 2 envies agent 3.
		\item $\{A_1 = [0, \frac{1}{12}], A_2 = [\frac{1}{12}, \frac34], A_3 = [\frac34, 1]\}$ is proportional, but not locally proportional since agent 1 envies their only neighbor, agent 2.
		\item $\{A_1 = [\frac14, \frac12], A_2 = [0, \frac14] \cup [\frac34, 1], A_3 = [\frac12, \frac34]\}$ satisfies none of the fairness properties, as agents 1 and 3 derive no value from their portions of the allocation, and are envious of agent 2.
	\end{itemize}

	We consider the computational task of finding a locally proportional allocation in the Robertson-Webb model \cite{RWModel}. In this model, an algorithm is allowed to make two kinds of queries. In an \eval query, the algorithm submits an agent $i$ and two real numbers $0 \leq x \leq y \leq 1$, and the query returns the value $v_i([x, y])$. In a \cut query, the algorithm submits an agent $i$ and real numbers $x, \alpha \in [0, 1]$, and the query returns a real number $y \in [0, 1]$ such that $v_i([x, y]) = \alpha$, if such a $y$ exists. The algorithm is allowed to perform arbitrary computation between queries, and complexity is measured only by the total number of queries. Since a finite algorithm will never be able to completely learn $v_i$ for any $i$, it is only correct if it outputs a \emph{provably} locally proportional allocation, in the sense that it is locally proportional with respect to any collection of valuation functions $\{v_i\}_{i \in V(G)}$ consistent with the sequence of query responses returned.
	
	\section{Proof of quadratic lower bound}\label{secProof}
	
	Given any positive integer $n$, we define an $n$-agent adversarial instance of the graphical cake-cutting problem as follows. Let $m$ be the greatest integer multiple of 32 that is strictly less than $n$, and let $r := n - m$. Let the underlying graph $G$ have a vertex set comprised of three disjoint parts: $V(G) := L \cup R \cup U$, where $\abs{L} = \abs{R} = \frac{m}{2}$ and $\abs{U} = r \geq 1$. Every vertex in $U$ is adjacent to every other vertex in $V(G)$, while the edges between $L$ and $R$ are determined randomly: for each agent $i \in L$, we randomly select a subset $S_i \subseteq R$ of size exactly $\frac{m}{4}$ and add edges between $i$ and each agent in $S_i$. The $S_i$ subsets are chosen uniformly from the set of all subsets of $R$ of size $\frac{m}{4}$, independently for each $i$. Throughout the execution of the cake-cutting algorithm, we define the response of every query to be consistent with every agent having a uniform density valuation function, i.e., for any $i$, $\eval_i(x, y) := y - x$ and $\cut_i(x, \alpha) := x + \alpha$ (assuming $x + \alpha \leq 1$).
	
	Note that the queries are answered non-adaptively and deterministically; the only random part of the construction concerns the structure of the graph. Throughout this paper, all probabilities are taken only over the space of possible graphs. Our main theorem is as follows.
	
	\begin{theorem}\label{thmMain}
		With probability tending to 1 as $n \to \infty$, this construction produces an instance on which it is not possible for any cake-cutting algorithm to output a provably locally proportional allocation after at most $8\left\lfloor\frac{n - 1}{32}\right\rfloor^2$ queries.
	\end{theorem}
	
	In particular, this means that, for large enough $n$, there is at least \emph{some} graph of $n$ vertices from this construction on which no cake-cutting algorithm succeeds in $8\left\lfloor\frac{n - 1}{32}\right\rfloor^2$ queries. This establishes the $\Omega(n^2)$ lower bound on the query complexity of local proportionality. We break the proof up into a series of four lemmas.

	\begin{lemma}\label{lemAllAgentsGetSameSizedCake}
		Suppose all $n$ agents have uniform density valuation functions, and $\{A_i\}_{i \in V(G)}$ is a locally proportional allocation. Then each agent $i$ receives a piece of cake of measure $\abs{A_i} = \frac1n$.
	\end{lemma}

	\begin{proof}
		The only fact we must use that is specific to our construction is that $G$ is connected, which is true since $\abs{U} \geq 1$. Suppose toward a contradiction that not all agents receive a piece of cake of measure $\frac1n$. Then there must be some pair of agents $i_1, j_1 \in V(G)$ such that $\abs{A_{i_1}} \neq \abs{A_{j_1}}$. Let $H$ be the directed graph with vertex set $V(H) := V(G)$, and an edge from $i$ to $j$ whenever $\{i, j\} \in E(G)$ and $\abs{A_i} < \abs{A_j}$. Observe that $H$ must contain at least one directed edge, for otherwise it is implied by induction along the path between $i_1$ and $j_1$ that $\abs{A_{i_1}} = \abs{A_{j_1}}$. Since $H$ is clearly acyclic, it follows that $H$ has at least one source, that is, a vertex $i_2$ with at least one outgoing edge to another vertex $j_2$, and no incoming edges. In other words, $\abs{A_{j_2}} > \abs{A_{i_2}}$, and for any $j \in N(i_2) \setminus \{j_2\}$, $\abs{A_{j}} \geq \abs{A_{i_2}}$. Then the average value agent $i_2$ has for neighbors' portions of the allocation is
		\begin{align*}
			\frac{1}{\deg(i_2)} \sum_{j \in N(i_2)} \abs{A_j} &= \frac{1}{\deg(i_2)} \left(\abs{A_{j_2}} + \sum_{j \in N(i_2) \setminus \{j_2\}} \abs{A_j}\right)\\
			&> \frac{1}{\deg(i_2)} \left(\abs{A_{i_2}} + \sum_{j \in N(i_2) \setminus \{j_2\}} \abs{A_{i_2}}\right)\\
			&= \abs{A_{i_2}}.
		\end{align*}
		Thus, we have shown that the allocation is not provably locally proportional for agent $i_2$, contradicting our assumption.
	\end{proof}

	We require some additional notation. Fix an allocation $\{A_i\}_{i \in V(G)}$. For any agent $i \in L$, let $B_i$ be the set of \emph{boundary points} of $\bigcup_{j \in S_i} A_j$, and let $b_i := \abs{B_i}$; i.e., $b_i$ is twice the number of disjoint intervals comprising the subset of $[0, 1]$ allocated to the neighbors of $i$.
	
	\begin{lemma}\label{lemMustQueryBoundary}
		Suppose a cake-cutting algorithm outputs a provably locally proportional allocation $\{A_i\}_{i \in V(G)}$ after at most $8\left\lfloor\frac{n - 1}{32}\right\rfloor^2$ queries. Then there must exist some set $M \subseteq L$ of size $\frac{m}{4}$ such that, for all $i \in M$, $b_i \leq \frac{m}{16}$.
	\end{lemma}
	
	\begin{proof}
		First, note that Lemma \ref{lemAllAgentsGetSameSizedCake} applies to the allocation returned by the algorithm, since this allocation must be locally proportional with respect to all agents having uniform valuation functions, as uniform valuation functions are consistent with all queries according to the construction. We will use this fact shortly.
		
		If the algorithm terminated after at most
		$$8\left\lfloor\frac{n - 1}{32}\right\rfloor^2 = \frac{\left(32\left\lfloor (n - 1) / 32 \right\rfloor\right)^2}{128} = \frac{m^2}{128}$$
		queries, then there are at most $\frac{m}{4}$ agents $i \in L$ whose valuation functions were queried at least $\frac{m}{32}$ times. Since $\abs{L} = \frac{m}{2}$, this means that there are \emph{at least} $\frac{m}{4}$ agents $i \in L$ whose valuation functions were queried \emph{less than} $\frac{m}{32}$ times. Let $M \subseteq L$ be a set of $\frac{m}{4}$ such agents. Suppose toward a contradiction that, for some $i \in M$, $b_i > \frac{m}{16}$, which means $b_i \geq \frac{m}{16} + 1$, since $m$ is divisible by 32. Let $Q_i \subseteq [0, 1]$ be the set of points that were submitted as part of an $\eval_i$ or $\cut_i$ query (but only as the $x$ in a $\cut_i$ query, not as the $\alpha$), or returned by a $\cut_i$ query. Since $v_i$ was queried at most $\frac{m}{32} - 1$ times, and each query involves 2 points in $[0, 1]$, there are at most $2(\frac{m}{32} - 1) = \frac{m}{16} - 2$ points in $Q_i$. As $\abs{B_i} \geq \abs{Q_i} + 3$, there are at least 3 points in $B_i \setminus Q_i$. At most 2 of those 3 points can be 0 or 1, so let $x$ be such a point in the interior of the unit interval.
		
		Let $\varepsilon > 0$ be the minimum distance between $x$ and any point in the set $Q_i \cup (B_i \setminus \{x\}) \cup \{0, 1\}$. Since $x \in B_i$, we may write $\{[x - \varepsilon, x], [x, x + \varepsilon]\} = \{I_2, I_0\}$, where $I_2 \subseteq \bigcup_{j \in S_i} A_j$ and $I_0$ only overlaps with $\bigcup_{j \in S_i} A_j$ at endpoint(s). Consider the valuation function $v_i$ which has uniform density, except over $I_0$, where it has density 0, and $I_2$, where it has density 2. Since the interior of $I_2 \cup I_0$ is disjoint from $Q_i$, and $v_i(I_2 \cup I_0) = \abs{I_2 \cup I_0}$, $v_i$ is consistent with all previous queries. However,
		\begin{align*}
			\frac{1}{\deg(i)} \sum_{j \in N(i)} v_i(A_j) &= \frac{1}{\deg(i)}\left(\sum_{j \in N(i)} \abs{A_j} + \varepsilon\right)\\
			&= \frac{1}{\deg(i)}\left(\sum_{j \in N(i)} \frac1n + \varepsilon\right) \stext{by Lemma \ref{lemAllAgentsGetSameSizedCake}}\\
			&= \frac1n + \frac{\varepsilon}{\deg(i)}\\
			&> \frac1n\\
			&= \abs{A_{i}} \stext{by Lemma \ref{lemAllAgentsGetSameSizedCake}}\\
			&\geq v_i(A_i).
		\end{align*}
		Thus, we have shown that the allocation is not provably locally proportional for agent $i$, contradicting our assumption.
	\end{proof}
	
	For any agent $i \in L$, we define a \emph{segmentation} of $i$ to be a partition $\mathcal{P} = \{\seq{R}{m/4}\}$ of $R$ such that each $R_j$ set has size 2 and at most $\frac14$ of the $R_j$ sets satisfy $\abs{S_i \cap R_j} = 1$. For any set $M \subseteq L$, we say that $\mathcal{P}$ is a segmentation of $M$ if it is a segmentation of every $i \in M$.
	
	\begin{lemma}\label{lemBoundaryImpliesPermutation}
		For any $M \subseteq L$, if there exists a locally proportional allocation $\{A_i\}_{i \in V(G)}$ such that, for all $i \in M$, $b_i \leq \frac{m}{16}$, then $M$ has a segmentation.
	\end{lemma}

	\begin{proof}
		We define an ordering $\preceq$ on the agents in $R$ as follows: $r \preceq r'$ whenever $\inf(A_r) \leq \inf(A_{r'})$, i.e., the first disjoint interval in the piece of cake allocated $r$ precedes the first disjoint interval in the piece of cake allocated to $r'$. Clearly, $\preceq$ is a well-defined total order on $R$ (by Lemma \ref{lemAllAgentsGetSameSizedCake}, each $A_i$ is nonempty, and by disjointness, the infimums of $A_i$ and $A_j$ must differ for $i \neq j$). Define $\seq{R}{m/4}$ by taking elements in the order of $\preceq$, i.e., $R_1$ contains the first two agents according to $\preceq$, $R_2$ contains the next two, and so on.
		
		To prove that $\{\seq{R}{m/4}\}$ is a segmentation of $M$, fix an arbitrary $i \in M$. Suppose that, for a given $j \in \{1, 2, \dots \frac{m}{4}\}$, $\abs{S_i \cap R_j} = 1$, which means that $R_j$ contains both a vertex $r_j \in R_j$ to which $i$ is adjacent and a vertex $r'_j \in R_j$ to which $i$ is not adjacent. In order to separate the first interval of $A_{r_j}$, which is in $B_i$, from the first interval of $A_{r'_j}$, which is not in $B_i$, there must be a boundary point somewhere between $\inf(A_{r_j})$ and $\inf(A_{r'_j})$. Moreover, if this happens for multiple values of $j$, each of these boundary points must be distinct, as successive $R_j$ sets contain pairs of agents allocated pieces of cake with strictly larger infimums. Since $b_i \leq \frac{m}{16}$, it follows that there are at most $\frac{m}{16}$ values of $j$ such that $\abs{S_i \cap R_j} = 1$, which is $\frac14$ of the $R_j$ sets.
	\end{proof}

	\begin{lemma}\label{lemPermutationProbabilityBound}
		Let $p_m$ denote the probability that there exists a set $M \subseteq L$ of size $\frac{m}{4}$ that has a segmentation. As $m \to \infty$, $p_m \to 0$.
	\end{lemma}
	
	\begin{proof}
		Let $\Pi$ denote the set of partitions of $R$ into sets of size 2. For any $M \subseteq L$ of size $\frac{m}{4}$ and any $\mathcal{P} \in \Pi$, let $E(M, \mathcal{P})$ be the event that $\mathcal{P}$ is a segmentation of $M$. Our strategy is to apply a union bound over all possible $M$ and $\mathcal{P}$.
		
		Thus, fix an arbitrary set $M \subseteq L$ of size $\frac{m}{4}$ and an arbitrary partition $\mathcal{P}$ of $R$ into $\frac{m}{4}$ sets $\seq{R}{m/4}$, each of size 2. For any $i \in M$ and $1 \leq j \leq \frac{m}{4}$, let
		$$X_{i, j} := \twocases{\txt{if $i$ is adjacent to at least one of the two vertices in $R_j$}}{1}{\txt{otherwise}}{0}.$$
		Note that, for any fixed $i \in M$, the $X_{i, j}$ are negatively associated across $j$, so Hoeffding's inequality applies to the sum of $X_{i, j}$ across $j$.\footnote{Formally, this follows since $X_{i, j}$ is the result of applying the monotone ``or'' function to the negatively-associated Fermi-Dirac distribution. See Dubhashi and Ranjan \cite[Propositions 7, 8-2, and 44-2]{NegativeAssociation}.} Since the expectation of each $X_{i,j}$ is clearly at least $\frac34$, we have
		\begin{align*}
			\Pr\left[\sum_{j = 1}^{m/4} X_{i, j} \leq \frac58 \cdot \frac{m}{4}\right] &= \Pr\left[\sum_{j = 1}^{m/4} X_{i, j} \leq \left(\frac34 \cdot \frac{m}{4}\right) - \frac{m}{32}\right]\\
			&\leq \Pr\left[\sum_{j = 1}^{m/4} X_{i, j} \leq \ee\left[\sum_{j = 1}^{m/4} X_{i, j}\right] - \frac{m}{32}\right]\\
			&\leq \exp\left(-\frac{2(m/32)^2}{m/4}\right) \stext{Hoeffding's inequality \cite{Hoeffding}}\\
			&= \exp\left(-\frac{m}{128}\right).
		\end{align*}
		In other words, for any arbitrary $i \in M$, the probability that at most $\frac58$ of the $R_j$ sets contain a vertex incident to $i$ is bounded by $\exp\left(-\frac{m}{128}\right)$. By the symmetry of the way edges are randomly chosen, it analogously holds that the probability that at most $\frac58$ of the $R_j$ sets contain a vertex \emph{not} incident to $i$ is also bounded by $\exp\left(-\frac{m}{128}\right)$. Hence, with probability at least $1 - 2\exp\left(-\frac{m}{128}\right)$, more than $\frac58$ of the $R_j$ sets contain a vertex incident to $i$ \emph{and} more than $\frac58$ of the $R_j$ sets contain a vertex \emph{not} incident to $i$. When this happens, it follows that more than $\frac14$ of the $R_j$ sets contain both a vertex to which $i$ is adjacent and a vertex to which $i$ is not adjacent, in which case $\mathcal{P}$ is not a segmentation for $i$. Thus, the probability that $\mathcal{P}$ is a segmentation for $i$ is at most $2\exp\left(-\frac{m}{128}\right)$.
		Moreover, whether the fixed partition $\mathcal{P}$ is a segmentation for $i$ depends only on the random choice of $S_i$. Since this is determined independently for each $i$ by construction, it follows that
		\begin{align*}
			\Pr[E(M, \mathcal{P})] &= \Pr\left[\mathcal{P}\txt{ is a segmentation for M}\right]\\
			&= \prod_{i \in M}\Pr\left[\mathcal{P}\txt{ is a segmentation for i}\right]\\
			&\leq \prod_{i \in M}2\exp\left(-\frac{m}{128}\right)\\
			&= 2^{m/4}\exp\left(-\frac{m^2}{512}\right).
		\end{align*}
		
		Therefore, we have
		\begin{align*}
		p_m &= \Pr\left[\bigcup_{\substack{M \subseteq L\\\abs{M} = m/4}} \bigcup_{\mathcal{P} \in \Pi} E(M, \mathcal{P})\right]\\
		&\leq \sum_{\substack{M \subseteq L\\\abs{M} = m/4}} \sum_{\mathcal{P} \in \Pi} \Pr\left[E(M, \mathcal{P})\right] \stext{by the union bound}\\
		&\leq \sum_{\substack{M \subseteq L\\\abs{M} = m/4}} \sum_{\mathcal{P} \in \Pi} 2^{m/4} \exp\left(-\frac{m^2}{512}\right)\\
		&= {m/2 \choose m/4} \cdot \left(\frac{m}{2} - 1\right)!! \cdot 2^{m/4} \cdot \exp\left(-\frac{m^2}{512}\right),
		\end{align*}
		where in the final equality we have used the easy calculation that $\abs{\Pi} = (\frac{m}{2} - 1)!!$, which is the product of all odd positive integers less than $\frac{m}{2}$.\footnote{To see this, pick any arbitrary vertex in $R$. There are $\frac{m}{2} - 1$ possible vertices that can be paired with it. Then pick another arbitrary vertex. Now there are $\frac{m}{2} - 3$ possible vertices to pair with, and so on. See: \url{http://oeis.org/A001147}} This entire expression vanishes as $m \to \infty$.
	\end{proof}

\begin{proof}[Proof of Theorem \ref{thmMain}]
Let $T$ be an arbitrary algorithm in the Robertson-Webb model. By Lemma \ref{lemMustQueryBoundary}, if $T$ outputs a provably locally proportional allocation on a given instance from our construction after at most $8\left\lfloor\frac{n - 1}{32}\right\rfloor^2$ queries, there exists an allocation (namely, the allocation output by $T$) and some set $M \subseteq L$ of size $\frac{m}{4}$ such that, for all $i \in M$, $b_i \leq \frac{m}{16}$. By Lemma \ref{lemBoundaryImpliesPermutation}, such a set $M$ always has a segmentation. Therefore, the probability that $T$ outputs a provably locally proportional allocation after at most $8\left\lfloor\frac{n - 1}{32}\right\rfloor^2$ queries is bounded by the probability that there exists a set $M \subseteq L$ of size $\frac{m}{4}$ that has a segmentation. As $n \to \infty$, we have $m \to \infty$ as well, so by Lemma \ref{lemPermutationProbabilityBound}, this probability vanishes. Hence, $T$ fails with high probability over the randomness in the construction of $G$.
\end{proof}

\section{Conclusion}\label{secConclusion}

We have shown that any cake-cutting algorithm for local proportionality requires $\Omega(n^2)$ queries. Just as with the $\Omega(n^2)$ lower bound for envy-freeness \cite{EFLowerBound}, the construction uses a simple, non-adaptive query adversary in which all queries are answered according to the same valuation function. Roughly, both proofs are based on the observation that any algorithm must verify the value that $\Omega(n)$ agents have for $\Omega(n)$ disjoint pieces of cake---though proving this necessity is substantially more involved for local proportionality than for envy-freeness.

Any further improvements to this lower bound would imply an improved lower bound for envy-freeness as well, which has been a major open question in fair division for years. A more promising direction for future work on local proportionality is to attempt to narrow the gap from above. For instance, can we find an algorithm for local proportionality making polynomially-many queries? This seems to be more plausible than the analogous question for envy-freeness.
	
	\section*{Acknowledgments}
	
	I would like to thank Ariel Procaccia and Daniel Halpern for helpful advice and feedback on an earlier version of this paper. This material is based upon work supported by the National Science Foundation Graduate Research Fellowship Program under Grant No.\txt{} DGE1745303. Any opinions, findings, and conclusions or recommendations expressed in this material are those of the author(s) and do not necessarily reflect the views of the National Science Foundation.
	
	\bibliographystyle{elsarticle-num}
	\bibliography{bibliography}
\end{document}